\documentclass[letterpaper,10pt,twocolumn]{IEEEtran} 

\IEEEoverridecommandlockouts                              

\usepackage{amsfonts}
\usepackage{float}
\usepackage[colorlinks,linkcolor=blue]{hyperref}
\usepackage[latin1]{inputenc}
\usepackage{amsmath}
\usepackage{graphicx}
\usepackage{amssymb}
\usepackage{amsthm}
\usepackage{verbatim}
\usepackage{epsfig}
\usepackage[labelfont=bf]{caption}
\usepackage{enumerate}
\usepackage{subfig}
\usepackage{epstopdf}

\newtheorem{definition}{Definition}

\newtheorem{remark}{Remark}

\newtheorem{proposition}{Proposition}

\def\begequarr{\begin{eqnarray}}
\def\endequarr{\end{eqnarray}}
\def\begequarrs{\begin{eqnarray*}}
\def\endequarrs{\end{eqnarray*}}
\def\begarr{\begin{array}}
\def\endarr{\end{array}}
\def\begequ{\begin{equation}}
\def\endequ{\end{equation}}
\def\lab{\label}
\def\begdes{\begin{description}}
\def\enddes{\end{description}}
\def\begenu{\begin{enumerate}}
\def\begite{\begin{itemize}}
\def\endite{\end{itemize}}
\def\endenu{\end{enumerate}}

\def\lef[{\left[\begin{array}}
\def\rig]{\end{array}\right]}
\def\qed{\hfill$\Box \Box \Box$}
\def\begcen{\begin{center}}
\def\endcen{\end{center}}
\def\begrem{\begin{remark}\rm}
\def\endrem{\end{remark}}

\def\begmat#1{\begin{bmatrix}#1\end{bmatrix}}
\def\begali#1{\begin{align}{#1}\end{align}}
\def\begalis#1{\begin{align*}{#1}\end{align*}}

\def\calk{{\cal K}}
\def\caly{{\cal Y}}

\def\calh{{\cal H}}

\def\call{{\cal L}}

\def\cald{{\cal D}}

\def\hatthe{\hat{\theta}}
\def\tilthe{\tilde{\theta}}
\def\dothatthe{\dot{\hat{\theta}}}

\def\liminf{\lim_{t \to \infty}}
\def\callinf{{\cal L}_\infty}
\def\litcallinf{\ell_\infty}
\def\L2e{{\cal L}_{2e}}
\def\bul{\noindent $\bullet\;$}
\def\rea{\mathbb{R}}
\def\intnum{\mathbb{Z}}

\def\adj{\mbox{adj}}

\def\hal{{1 \over 2}}
\def\et{\varepsilon_t}

\def\IJACSP{{\it Int. J. on Adaptive Control and Signal Processing}}

\def\TAC{{\it IEEE Trans. Automatic Control}}

\def\AUT{{\it Automatica}}

\usepackage{color}

\graphicspath{{figs/}}

\title{\LARGE \bf
New Results on Parameter Estimation via Dynamic Regressor Extension and Mixing: Continuous and Discrete-time Cases
}

\author{
Romeo Ortega, {\it Fellow, IEEE},
Stanislav Aranovskiy, {\it Senior member, IEEE},
Anton A. Pyrkin, {\it Member, IEEE}, 
Alessandro Astolfi, {\it Fellow, IEEE},
Alexey A. Bobtsov, {\it Senior member, IEEE} 
\thanks{R. Ortega is with Laboratoire des Signaux et Syst\`emes, CNRS-SUPELEC, Plateau du Moulon, 91192, Gif-sur-Yvette, France and ITMO University, Kronverkskiy av. 49, St. Petersburg, 197101, Russia.}
\thanks{S. Aranovskiy is with CentraleSup\'elec -- IETR, Avenue de la Boulaie, 35576 Cesson-S\'evign\'e, France}
\thanks{ A. Pyrkin is with Hangzhou Dianzi University, Hangzhou, 310018, China.}
\thanks{S. Aranovskiy, A. Pyrkin and A. Bobtsov are with the Faculty of Control Systems and Robotics, ITMO University, Kronverkskiy av. 49, St. Petersburg, 197101, Russia.} 
\thanks{A. Astolfi is with the Department of Electrical and Electronic Engineering, Imperial College London, London SW7 2AZ, UK and the DICII,
Universita di Roma ``Tor Vergata'', Via del Politecnico 1, 00133 Roma, Italy}
\thanks{A. Pyrkin is a corresponding author. E-mail:  {\tt\small a.pyrkin@gmail.com}}}

\begin{document}

\maketitle
\thispagestyle{plain}
\pagestyle{plain}

\begin{abstract}
We present some new results on the dynamic regressor extension and mixing parameter estimators for linear regression models recently proposed in the literature. This technique has proven instrumental in the solution of several open problems in system identification and adaptive control. The new results include: (i) a unified treatment of the continuous and the discrete-time cases; (ii) the proposal of two new extended regressor matrices, one which guarantees a quantifiable {\em transient performance improvement}, and the other exponential convergence under conditions that are {\em strictly weaker} than regressor persistence of excitation; and (iii) an alternative estimator ensuring parameter estimation in {\em finite-time} that retains its alertness to track time-varying parameters. Simulations that illustrate our results are also presented.
\end{abstract}

\section{Introduction}
\lab{sec1}
%
Estimation of the parameters that describe an underlying physical setting is one of the central problems in control and systems theory that has attracted the attention of many researchers for several  years. A typical scenario, which appears in system identification and adaptive control \cite{GOOSIN,IOASUN,LJU,NARANN,SASBOD}, is when the unknown parameters and the measured data are linearly related in a so-called {\em linear regression equation} (LRE). Classical solutions for this problem are gradient and least-squares (LS) estimators. The main drawback of these schemes is that convergence of the parameter estimates relies on the availability of signal excitation, a feature that is codified in the restrictive assumption of persistency of excitation (PE) of the regressor vector. Moreover, their transient performance is highly unpredictable and only a weak monotonicity property of the estimation errors can be guaranteed.  

To overcome these two problems a new parameter estimation procedure, called dynamic regressor extension and mixing (DREM), has recently been proposed in \cite{ARAetaltac} for continuous-time (CT) and in \cite{BELetalsysid} for discrete-time (DT) systems. The construction of DREM estimators proceeds in two steps, first, the inclusion of a {\em free linear operator} that creates an extended, matrix LRE. Second, a {\em nonlinear} manipulation of the data that allow generating, out of an $m$-dimensional LRE,  $m$ {\em scalar}, and independent, LRE. DREM estimators have been successfully applied in a variety of identification and adaptive control problems. Interestingly, it has been shown in \cite{ORTetalaut} that DREM can be reformulated as a functional Luenberger observer.

DREM estimators outperform classical gradient or LS estimators in the following precise aspects: independently of the excitation conditions, DREM guarantees monotonicity of {\em each} element of the parameter error vector that is much stronger than monotonicity of the vector {\em norm}, which is ensured with classical estimators. Moreover, parameter convergence in DREM is established {\em without} the PE condition. Instead of PE a non-square integrability condition on the determinant of a designer-dependent extended regressor matrix is imposed. A final interesting property of DREM that has been established in \cite{GERetalsysid} is that it can be used to generate estimates with {\em finite-time convergence} (FTC), under interval excitation assumption.

The following new results on DREM are presented here:

\noindent {\bf (i)} The unified treatment of the CT and the DT cases.

\noindent {\bf (ii)} The definition of new linear operators that:

\bul ensure parameter error  convergence under excitation conditions that are {strictly weaker} than regressor PE;

\bul guarantee a {transient performance improvement};

\bul show that DREM contains, as a {particular case}, the extended LRE proposed in \cite{KRE}, which is used also in the adaptive controllers recently proposed in \cite{CHOetal,CHOetalijacsp,ROYetal}.

\noindent {\bf (iii)} An alternative estimator, {ensuring FTC}, that retains its alertness to track time-varying parameters.

The remainder of the paper is organized as follows. To set up the notation a brief  description of gradient and DREM estimators is given in Section \ref{sec2}.  In Section \ref{sec3} we present the new version of DREM that ensures convergence under excitation conditions that are strictly weaker than regressor PE. In Section \ref{sec4} a general form of the free operator used in DREM is proposed to, on one hand, re-derive the extended regressor of \cite{KRE} and, on the other hand, prove that transient performance is---quantifiably---improved.  Section \ref{sec5} presents a new DREM-based estimator with FTC. Simulation results are presented in Section \ref{sec6}. The paper is wrapped-up with future research in  Section \ref{sec7}.\\

\noindent {\bf Notation.} $I_n$ is the $n \times n$ identity matrix. $\rea_{>0}$, $\rea_{\geq 0}$, $\intnum_{>0}$ and $\intnum_{\geq 0}$ denote the positive and non-negative real and integer numbers, respectively. For $x \in \rea^n$, we denote $|x|^2:=x^\top x$. Continuous-time (CT) signals $s:\rea_{\geq 0} \to \rea$ are denoted $s(t)$, while for discrete-time (DT) sequences $s:\intnum_{\geq 0} \to \rea$ we use $s(k):=s(kT_s)$, with $T_s \in \rea_{> 0}$ the sampling time. The action of an operator $\mathcal H:\callinf \to \callinf$ on a CT signal $u(t)$ is denoted $\mathcal H[u](t)$, while for an operator $\calh:\litcallinf \to \litcallinf$ and a sequence $u(k)$ we use  $\mathcal H[u](k)$. When a formula is applicable to CT signals and DT sequences the time argument is omitted. 

%
\section{Background Material}
\lab{sec2}
%
We deal with the problem of on-line estimation of the unknown, constant parameters $\theta \in \rea^m$ appearing in a LRE of the form
\begequ
\lab{lre}
y=\phi^\top\theta + \et
\endequ 
where $y \in \rea$ and $\phi \in \rea^m$ are {\em measurable} CT or DT signals and $\et$ is a (generic) exponentially decaying signal.\footnote{This signal may be stemming from the effect of the initial conditions of various filters used to generate the LRE.}  It is well-known that the availability of a LRE of the form \eqref{lre} is instrumental for the development of most system identifiers and adaptive controllers \cite{SASBOD}. Following standard practice, throughout the paper, the term $\et$ is omitted.
\subsection{Gradient estimator and the PE condition}
\lab{subsec21}
%
In this subsection we recall the well-known gradient estimator, derive its parameter error equation (PEE) and recall its stability properties. Although this material is very well-known, it is included to make the document self-contained and set up the notation.  First, we introduce the following.

\begin{definition} \em
\lab{def1} 
A bounded signal $\phi\in \rea^{m}$ is PE (denoted $\phi \in PE$) if there exist $\alpha \in \rea_{>0}$ such that
$$
\int_t^{t+T}\phi(\tau)\phi^\top(\tau) d\tau \geq \alpha I_m,\;\forall t \in \rea_{\geq 0},
$$ 
for some $T \in \rea_{> 0}$ in CT or
$$  
\sum_{j=k+1}^{k + K} \phi(j) \phi^\top(j) \geq \alpha I_m,\;\forall k \in \intnum_{\geq 0},
$$
for some $K \in \intnum_{> 0}$, with $K \geq m$, in DT.
\qed
\end{definition}

The following proposition is a milestone for systems theory and may be found in all identification and adaptive control textbooks, {\em e.g.}, \cite{SASBOD}.

\begin{proposition}\em
\lab{pro1}
Consider the LRE  \eqref{lre}.\\

\noindent {\bf (CT)} The CT gradient-descent estimator
\begequ
\lab{graest}
\dothatthe(t)=\gamma \phi(t)[ y(t)-\phi^\top(t) \hatthe(t)],
\endequ
with $\gamma>0$ ensures the following.

\bul The {\em norm} of the parameter error vector $\tilthe:=\hatthe-\theta$ is monotonically non-increasing, that is,
\begequ
\lab{normonct}
|\tilde\theta(t_b)| \leq |\tilde\theta(t_a)|,\;\forall t_b \geq t_a \in \rea_{\geq 0}. 
\endequ
\bul The CT PEE is given by
$$
\dot{\tilde{\theta}}(t)=-\gamma\phi(t)\phi^{\top}(t)\tilde\theta(t),
$$ 
and its zero equilibrium is globally {exponentially} stable (GES) {\em if and only if} $\phi(t) \in PE$. Moreover, there exist an optimal value of $\gamma$ for which the rate of convergence is {\em maximum}.

\noindent {\bf (DT)} The DT gradient-descent estimator
$$
\hat{\theta}(k) = \hat{\theta}(k-1) + {\phi(k) \over \gamma + |\phi(k)|^2} [y(k) - \phi^\top (k) \hat{\theta}(k-1)],
$$
ensures the following.

\bul The {\em norm} of the parameter error vector verifies
\begequ
\lab{normondt}
|\tilde{\theta}(k_b)| \leq  |\tilde{\theta}(k_a)|,\;\forall k_b \geq k_a \in \intnum_{\geq 0}. 
\endequ
\bul The DT PEE is given by
$$
\tilde{\theta}(k) =\bigg[I_m- {1 \over \gamma + |\phi(k)|^2}\phi(k) \phi^\top (k)\bigg] \tilde{\theta}(k-1) ,
$$ 
and its zero equilibrium is GES {\em if and only if}  $\phi(k) \in PE$.

\qed
\end{proposition}

In most applications, PE is an extremely restrictive condition, hence the interest of relaxing it.  See \cite{ORTetaltac} for a recent review of new estimators relaxing the PE condition, which include the ones reported in \cite{CHOetal,CHOetalijacsp,ROYetal}.
%
\subsection{Generation of $m$ scalar LRE via DREM}
\lab{subsec22}
%
To overcome the limitation imposed by the PE condition and improve the transient performance of the estimator the DREM procedure, introduced in \cite{ARAetaltac,BELetalsysid}, generates $m$ new, one--dimensional, LRE to  {\em independently} estimate each of the parameters. The first step in DREM is to introduce a {\em linear, single-input $m$-output, bounded-input bounded-output (BIBO)--stable} operator $\calh$ and define the vector $Y \in \rea^m$ and the matrix $\Phi \in \rea^{m \times m}$ as
\begequ
\label{yphi}
\begin{aligned}
Y & := \calh[y],\;\Phi :=\calh[\phi^\top].
\end{aligned}
\endequ
Clearly, because of linearity and BIBO stability, these signals satisfy
\begequ
\label{extlre}
Y = \Phi \theta.
\endequ 
At this point the key step of regressor ``mixing" of the DREM procedure is used to obtain a set of $m$ {\em scalar} equations as follows. First, recall that, for any (possibly singular) $m \times m$ matrix $M$ we have \cite{LANTIS} $\adj\{M\} M=\det\{M\}I_m$, where $\adj\{\cdot\}$ is the adjunct (also called ``adjugate") matrix. Now, multiplying from the left the vector equation \eqref{extlre} by the {\em adjunct matrix} of $\Phi$, we get
\begequ
\label{scalre}
\caly_i = \Delta \theta_i,\quad i \in \{1,2,\dots,m\}
\endequ
where we have defined the scalar function $\Delta \in \rea$
\begequ
\label{del}
\Delta :=\det \{\Phi\},
\endequ
and the vector $\caly \in \rea^m$
\begequ
\label{caly}
\caly := \adj\{\Phi\} Y.
\endequ

\begrem
\lab{rem1}
In \cite{LIO} an extended regressor like \eqref{extlre} has been constructed in CT using linear time-invariant (LTI) filters in the operator $\calh$ used in \eqref{yphi}---see also \cite{KRE}, where this modification is also discussed. Unfortunately, besides some simulation evidence, no quantitative advantage---with respect to the gradient estimation---has been established for it.
\endrem
\subsection{Properties of gradient parameter estimators in DREM}
\lab{subsec23}
%
The availability of the scalar LREs \eqref{scalre} is the main feature of DREM that distinguishes it with respect to all other estimators. Indeed, as shown in the propostion below---the proof of which may be found in \cite{ARAetaltac,BELetalsysid}---it allows obtaining significantly stronger results using simple gradient estimators.

\begin{proposition}\em
\lab{pro2}
Consider the scalar LREs  \eqref{scalre}.\\

\noindent {\bf (CT)} The CT gradient-descent estimators\footnote{In the sequel, the quantifier $i \in \{1,2,\dots,m\}$ is omitted for brevity.} 
\begequ
\label{ctest}
\dot{\hat{\theta}}_i(t) = \gamma_i\Delta(t)[\caly_i(t) - \Delta(t) \hat\theta_i(t)],
\endequ
with $\gamma_i \in \rea_{>0}$ ensures the following.

\bul The CT PEEs are given by
\begequ
\lab{peect}
\dot {\tilde \theta}_i(t)=-\gamma_i \Delta^2(t) \tilde\theta_i(t).
\endequ

\bul The {\em individual} parameter errors are monotonically non-increasing, that is,
$$
|\tilde\theta_i(t_b)| \leq |\tilde\theta_i(t_a)|,\;\forall t_b \geq t_a \in \rea_{\geq 0}. 
$$
\bul The following {\em equivalence} holds
$$
\liminf \tilde \theta_i(t)=0~~\Leftrightarrow~~\Delta(t) \notin \call_2,
$$ 
and convergence can be made {\em arbitrarily fast} increasing $\gamma_i$.

\bul If $\Delta(t) \in PE$, the convergence is {\em exponential}.\\

\noindent {\bf (DT)} The DT gradient-descent estimator
\begequ
\label{dtest}
\hat{\theta}_i(k) = \hat{\theta}_i(k-1) + {\Delta(k) \over \gamma_i + \Delta^2(k)} [\caly_i(k) - \Delta(k) \hat{\theta}_i(k-1)],
\endequ
ensures the following.

\bul The DT PEEs are given by
\begequ
\lab{peedt}
{\tilde \theta}_i(k)={1 \over 1 +{\Delta^2(k) \over  \gamma_i}} \tilde\theta_i(k-1).
\endequ
\bul The elements of the parameter error vector verify
\begequ
\lab{normondrem}
|\tilde{\theta}_i(k_b)| \leq  |\tilde{\theta}_i(k_a)|,\;\forall k_b \geq k_a \in \intnum_{\geq 0}. 
\endequ

\bul The following {\em equivalence} holds
$$
\liminf \tilde \theta_i(k)=0~~\Leftrightarrow~~\Delta(k) \notin \ell_2,
$$ 
and convergence can be made {\em arbitrarily fast} decreasing $\gamma_i$.

\bul If $\Delta(k) \in PE$, the convergence is {\em exponential}.
\qed
\end{proposition}

There are three important advantages of DREM over the standard gradient estimator.

\noindent {\bf P1} As shown in \eqref{normondrem} the {\em individual} parameter errors are monotonically non-increasing, a property that is {\em strictly stronger} than monotonicity of their norm indicated in \eqref{normonct} and \eqref{normondt}. 

\noindent {\bf P2}  Parameter convergence is established {\em without} the restrictive PE assumption---being replaced, instead, by a non square-integrability/summability assumption.  

\noindent {\bf P3} Convergence rates of DREM can be made {\em arbitrarily fast} simply increasing $\gamma_i$ in CT (or decreasing it in DT). 

\begrem
\lab{rem2}
Regarding the property {\bf P2}, in \cite{ARAetaltac} the relationship in CT between the conditions $\phi(t) \in PE$ and $\Delta(t) \notin \call_2$ is thoroughly discussed. In particular, in \cite{ARAetaltac} it has been shown that, for {\em arbitrary} regressor vectors $\phi(t)$, these conditions are unrelated. On the other hand, for the case of identification of LTI systems, it has been shown in \cite{BELetal} that $\phi(t) \in PE$ if and only if $\Delta(t) \in PE$ for {\em almost all} LTI operators $\calh$.
\endrem
%
\section{A DREM Estimator with Strictly Weaker Convergence Conditions}
\lab{sec3}
%
In this section we present a particular version of DREM for which it is possible to show that its convergence conditions are {\em strictly weaker} than $\phi \in PE$. Since the construction, and the results, are very similar for CT and DT estimators, for brevity, we consider the latter case only.
\begin{proposition}\em
\lab{pro3}
Consider the DT version of the LRE \eqref{lre}. Fix an integer $\bar K \geq m$ and define  \eqref{yphi} using the LTV operator 
$$
\calh:=\begmat{\phi(k-1)& \phi(k-2) & \cdots \phi(k-\bar K)}\begmat{q^{-1}\\ q^{-2} \\ \vdots \\ q^{-\bar K}}.
$$ 
Assume $\phi(k) \in PE$ and $\bar K \geq K$, with $K$ the size of the window given in Definition \ref{def1}. The scalar, gradient-descent DT estimators \eqref{dtest}, with $\Delta(k)$ and $\caly(k)$ defined in \eqref{del} and \eqref{caly}, ensure the following additional properties.

\bul  The condition for parameter convergence of DREM, {\em i.e.}, $\Delta(k) \not \in \ell_2$, is {\em strictly weaker} than $\phi(k) \in PE$. More precisely, the following implications hold:
\begali{
\lab{phipeimpdelnotl2}
\phi(k)  \in PE & \;\Rightarrow\;\Delta(k) \not \in \ell_2,\\
\lab{delnotl2notimpphipe}
\Delta(k) \not \in \ell_2\; & \not \Rightarrow \phi(k) \in PE.
}
\bul The condition for {\em exponential} parameter convergence of DREM, {\em i.e.}, $\Delta(k) \in PE$, is also weaker than $\phi(k) \in PE$ in the following precise sense
\begali{
\lab{phipeimpdelpe}
\phi(k)  \in PE & \;\Rightarrow\;\Delta(k) \in PE,\\
\lab{delpenotimpphipe}
\Delta(k)  \in PE\; [K \geq 2] & \not \Rightarrow \phi(k) \in PE\; [K \leq \bar K]
}
\end{proposition} 
\begin{proof}
To prove the claims we make the key observation that
\begequ
\lab{keyide}
\Phi(k)=\sum_{j=k+1}^{k + \bar K} \phi(j-(1+\bar K)) \phi^\top(j-(1+\bar K)).
\endequ 
The implications \eqref{phipeimpdelnotl2} and  \eqref{phipeimpdelpe} follow using the identity \eqref{keyide}, Definition \ref{def1} and noting the obvious fact that if $\phi(k) \in PE$ in a window of size $K$, then it is also PE for any window of size $\bar K \geq K$.

The proof of \eqref{delnotl2notimpphipe} is established with the following scalar counterexample: $\phi(k)=(k+1)^{-\frac{1}{4}}$ with $\bar{K}=1$. Since $\phi(k)$ tends to zero it is not PE, however, $\Delta(k) = (k+1)^{-\hal} \notin \ell_2.$

Finally, the proof of \eqref{delpenotimpphipe} is established with the following chain of implications: 
\begalis{
&\Delta(k) \in PE\; [\mbox{with\;}K \geq 2]  \Leftrightarrow\;\sum_{j=k+1}^{k + K}  \Delta^2(j) > 0,\;\forall k \in \intnum_{\geq 0}\\
                 & \Leftrightarrow\;\sum_{j=k+1}^{k + K} \prod_{i=1}^{m}\lambda^2_i\{\Phi(j)\} > 0,\;\forall k \in \intnum_{\geq 0}\\
                  & \Leftrightarrow \;\prod_{i=1}^{m}\lambda^2_i\{\Phi(k\!+\!1)\}\!+\!\dots\!+ \!\prod_{i=1}^{m}\lambda^2_i\{\Phi(k\!+\!K) \} > 0,\forall k \in \intnum_{\geq 0}\\
                    & \not \Rightarrow\; \lambda_i\{\Phi(k)\} > 0, \forall i \in \{1,\dots,m\},\;\forall k \in \intnum_{\geq 0}\\
                   & \Leftrightarrow\;\Phi(k) > 0,\;\forall k \in \intnum_{\geq 0}\\
                  & \Leftrightarrow\;\phi(k) \in PE\; [\mbox{with\;}K \leq \bar K],
}
where $\lambda_i\{\cdot\}$ denotes eigenvalues and in the third implication we have used the fact that $K>1$.
\end{proof}

\begrem
\lab{rem4}
The qualifiers $K \geq 2$ and $K \leq \bar K$ in \eqref{delpenotimpphipe} are necessary to complete the proof. Actually, it can be shown that, without these qualifiers,  [$\Delta(k)  \in PE \Rightarrow \phi(k) \in PE$]. 
\endrem
%
\section{Some Specific Choices of the Operator $\calh$}
\label{sec4}
%
In the reported literature of DREM we have considered the use of simple first-order, LTI filters or pure delays in the vector operator $\calh$, see \cite{ORTetalaut} for a discussion on LTV operators. One of the main contributions of the paper is to propose a general form for these operators and give an {\em explicit} choice, that ensures a quantifiable transient performance improvement of the estimator. Another advantage of these general operators is that, as a particular case, we obtain the extended LRE proposed in \cite{KRE} for adaptive state observation---referred in the sequel as Kreisselmeier's regressor extension (KRE).
%
\subsection{A general LTV operator $\calh$}
\lab{subsec41}
%
In this subsection we propose to generate $Y$ and $\Phi$ in \eqref{yphi} using, as elements of the operator $\calh$, the single-input single-output (SISO) LTV operators
\begequ
\lab{genhi}
\calh_i:=c_i^\top (\mathfrak{d}I_{n_i} - A_i)^{-1}b_i + d_i + \mu_i \cald_{i},
\endequ
where $A_i \in \rea^{n_i \times n_i}$, $b_i,c_i \in \rea^{n_i}$, $d_i,\mu_i \in \rea$ are {\em time-varying}, $n_i \in \intnum_{\geq 0}$ and the action of the operators $\mathfrak{d}$ and $\cald_{i}$ is defined as
$$
\mathfrak{d}[u]=\left\{ \begarr{ccl} {d u(t)\over dt}=:p[u](t) & \mbox{in} & CT \\ &&\\ u(k+1)=:q[u](k) & \mbox{in} & DT \endarr \right.
$$
and 
$$
\cald_{i}[u]=\left\{ \begarr{ccl} u(t-T_i),\;T_i \in \rea_{\geq 0} & \mbox{in} & CT \\ &&\\ q^{-K_i}u(k),\;K_i \in \intnum_{\geq 0}  & \mbox{in} & DT, \endarr \right.
$$
respectively. The triplets $(A_i,b_i,c_i)$ should define BIBO stable systems and all matrices are bounded.

The state-space realizations of the SISO, BIBO stable subsystems $z=\calh_i[u]$ are, clearly, given as
\begalis{
\dot x_i(t) & = A_i(t)x_i(t)+b_i(t)u(t) \\ z(t) & =c_i^\top (t) x_i(t)+d_i(t)u(t)+\mu_i(t)u(t-T_i),
}
in CT, and
\begalis{
x_i(k+1) & = A_i(k)x_i(k)+b_i(k)u(k) \\
z(k) & =c_i^\top (k) x_i(k)+d_i(k)u(k)+\mu_i(k)u(k-K_i),
}
in DT, with $x_i \in \rea^{n_i}$ the corresponding state. In view of the equivalence between GES and BIBO-stability for LTV systems with bounded realization matrices, these state-space systems are GES.

\begrem
\lab{rem6}
The LTV operators \eqref{genhi} are a generalization of the first order LTI ones or simple delays considered in the reported literature of DREM. LTV operators are also considered in \cite{ORTetalaut} to give a Luenberger observer interpretation of DREM.
\endrem
\subsection{Kreisselmeier's regressor extension}
\lab{subsec42}
%
The construction of the KRE  of \cite{ORTetaltac} proceeds as follows. Premultiplying \eqref{lre} by $\phi$ we obtain
$$
\phi y = \phi \phi^\top \theta,
$$
to which we can apply a {\em SISO}, linear, BIBO-stable operator $\calk$ to obtain the new, {\em matrix} LRE
\begequ
\lab{zomethe}
Z=\Omega \theta,
\endequ
where we have defined
\begali{
Z:=\calk[\phi y] \in \rea^m,\;\Omega :=\calk[\phi \phi^\top] \in \rea^{m \times m}.
\label{om}
}
Comparing \eqref{yphi}, \eqref{extlre} with  \eqref{zomethe}, \eqref{om} we see that the difference between DRE and KRE is that, in the first case, $Y$ and $\Phi$ are obtained filtering---with $m$ different filters---$y$ and $\phi$, and piling-up the filtered signals, while in the latter we filter $\phi y$ and $\phi \phi^\top $ with one filter. 

The proposition below, the proof of which is obtained via a direct calculation, shows that KRE is a particular case of the DRE construction using the generalized operators \eqref{genhi}.\footnote{The first author thanks Bowen Yi for bringing this fact to his attention.} 

\begin{proposition}\em
\lab{pro4}
Define  \eqref{yphi} using  \eqref{genhi} with
$$
n_i=1,\;c_i=1,\;A_i=-a,\;b_i=\phi_i,\;d_i=0,\;\mu_i=0.
$$  
Then, $Z=Y$ and $\Omega=\Phi$ as defined in \eqref{om}.
\qed
\end{proposition}

\begrem
\lab{rem7}
The KRE construction was first proposed by Kreisselmeier in  \cite{KRE} for CT systems and the particular case 
\begequ
\lab{calk}
\calk(p)={ 1 \over p+a},\;a>0, 
\endequ
the state-space realization of which is
\begali{
\nonumber
\dot{\Omega}(t)&=-a \Omega(t)+\phi(t)\phi^{\top}(t),\\
\lab{dotz}
\dot{Z}(t)&=-a Z(t)+\phi(t) y(t).
}
The LRE \eqref{zomethe} is used in the recently proposed MRACs \cite{CHOetal,ROYetal}, see also \cite{ORTetaltac} for a survey of the literature. 
\endrem
\subsection{An operator $\calh$ with guaranteed transient performance improvement}
\lab{subsec43}
%
One important feature of DREM is that it is possible to get the explicit solution of the PEEs, fully characterizing the time evolution of the parameter errors. Indeed, for the CT PEE \eqref{peect} we have
\begali{
\lab{solctpee}
\tilde \theta_i(t) = e^{-\gamma_i \int_0^t \Delta^2(s)ds} \tilde \theta_i(0),
}
Similarly, for the DT PEE \eqref{peedt} we have
\begali{
\lab{soldtpee}
\tilde \theta_i(k) = \prod_{j=0}^k \bigg[{1 \over 1 +{\Delta^2(j) \over \gamma_i}}\bigg] \tilde\theta_i(0).
}
As seen from the two previous equations the transient performance of the DREM estimators is univocally determined by the ``size" of $\Delta^2$---with a faster convergence obtained with a ``larger" $\Delta^2$. To improve the transient behavior of the DREM estimator we propose in this subsection a particular selection of the feedforward gains $d_i$ in the LTV operators $\calh_i$ given in \eqref{genhi}. Since the result is the same for CT and DT estimators, for brevity, we consider below the former case only.

To streamline the presentation of the result we define the matrix
\begequ
\lab{phi0}
\Phi_0(t):=\begmat{ c_1^\top(t) [ pI_{n_1} - A_1(t)]^{-1} b_1(t) +  \mu_1(t) \cald_{1}\\ \vdots \\ c_n^\top(t) [ pI_{n_n} - A_n(t)]^{-1} b_n(t) +  \mu_n(t) \cald_{n} }[\phi^\top ](t).
\endequ
That is, the CT extended regressor matrix \eqref{extlre} generated with the operators \eqref{genhi} {\em with $d_i(t)=0$}. 

\begin{proposition}
\lab{pro5}\em
Consider the CT DREM estimator (\ref{ctest}) with the LRE \eqref{extlre} generated with the operators \eqref{genhi}. Denote by $\tilde \theta^0_i(t)$ the parameter errors corresponding to the choice of $d_i(t)=0$ and $\tilde \theta^{\tt N}_i(t)$ those corresponding to 
\begequ
\lab{d}
d(t)=\adj\{\Phi_0(t)\}\phi(t),
\endequ
with $\Phi_0(t)$ defined in \eqref{phi0}, all the remaining parameters of $\calh_i$ and the estimators initial conditions the same for both cases. Then
$$
|\tilde \theta^0_i(t)| > |\tilde \theta^{\tt N}_i(t)|,\;\forall t \in \rea_{\geq 0}.
$$
\end{proposition}

\begin{proof}
From the definitions of  $\Phi(t)$ in \eqref{extlre}, the operators \eqref{genhi} and $\Phi_0(t)$ in \eqref{phi0} we have that
\begequ
\lab{phisum}
\Phi(t)=\Phi_0(t)+d(t)\phi^\top (t).
\endequ
In view of \eqref{solctpee} the proof is completed showing that 
$$
\det\{\Phi_0(t)\} < \det\{\Phi(t)\}.
$$
For, we apply Sylvester's determinant formula \cite{LANTIS} to \eqref{phisum} to get
\begalis{
\det\{\Phi(t)\}&=\det\{\Phi_0(t)\}+d^\top (t)\adj\{\Phi_0(t)\}\phi(t)\\
&=\det\{\Phi_0(t)\}\!+\!\phi^\top (t)\big[\adj\{\Phi_0(t)\}\big]\!^\top \!\adj\{\Phi_0(t)\}\phi(t)\\
&=\det\{\Phi_0(t)\} +|\adj\{\Phi_0(t)\}\phi(t)|^2,
}
where we have used \eqref{d} to obtain the second equation. The proof is completed noting that $\det\{\Phi_0(t)\}\neq 0$ implies that $\adj\{\Phi_0(t)\}$ is full rank. Hence, if $\phi(t)\neq 0$, the second right hand term of the last identity above is positive.
\end{proof}
%
\section{CT DREM Estimators with Alert Finite-Time Convergence}
\lab{sec5}
%
In \cite{GERetalsysid} we have showed that CT DREM can be used to generate estimates that converge in {\em finite time} under the weakest interval excitation assumption. 
\subsection{An FTC DREM}
\lab{subsec50}
%
For ease of reference, we recall the FTC result used in \cite{GERetalsysid} to solve a mutivariable adaptive control problem. 

\begin{proposition}\em
\lab{pro60}
Consider the scalar CT LREs  \eqref{scalre} and the gradient-descent estimator \eqref{ctest}. Fix a constant $\mu_i \in (0,1)$ and assume there exists a time $t_c  \in \rea_{> 0}$ such that
\begequ
\lab{inttc}
\gamma_i \int_0^{t_c} \Delta^2(s) ds \geq -\ln(\mu_i).
\endequ
Define the FTC estimate
\begequ
\lab{hatthew}
\hat \theta_i^{\tt FTC}(t) :={1  \over 1 -w^{\tt c}_i(t)}[ \hat \theta_i(t) - w_i^{\tt c}(t)\hat \theta_i(0)],
\endequ
where $w_i^{\tt c}(t)$ is defined via the clipping function
\begequ
\lab{condit}
w_i^{\tt c}(t)=\left \{ \begin{array}{lcr} \mu_i  \;\;\; \mbox{if} \;\;\;\;\;w_i(t) \geq \mu_i\\ \\ w_i(t)\;\; \mbox{if} \;\;\;\;\; w_i(t) <\mu_i,
\end{array} \right.
\endequ
with $w_i(t)$ given by 
\begequ
\lab{wi}
\dot{w}_i(t)=-\gamma\Delta^2(t) w_i(t),\;w_i(0) =1.
\endequ
The parameter estimation error converges to zero in {\em finite-time}. More precisely, $\hat \theta^{\tt FTC}_i(t) = \theta_i,\; \forall t \geq t_c.$
\end{proposition}

\begin{proof}
First, notice that the solution of \eqref{wi} is
$$
w_i(t)=e^{-\gamma_i \int_0^{t} \Delta^2(s)ds}.
$$
The key observation is that, using the equation above in \eqref{solctpee}, and rearranging terms we get that 
\begequ
\lab{keyrel}
[1-w_i(t)]\theta_i= \hat \theta_i(t) - w_i(t) \hat \theta_i(0).
\endequ
Now, observe that $w_i(t)$ is a non-increasing function and, under the interval excitation assumption \eqref{inttc}, we have that 
$$
w^{\tt c}_i(t)=w_i(t) < \mu_i,\;\forall t \geq t_c,
$$
completing the proof.
\end{proof}

\begrem
\lab{rem10}
The FTC property established in this section is ``trajectory-dependent'', in the sense that it relates only to the trajectory generated for the initial condition $w_i(0)=1$. This means that the flow of the closed-loop system contains other  trajectories, and the appearance of a perturbation may drive our ``good'' trajectory towards a ``bad'' one. This is, of course, a robustness problem that needs to be further investigated.
\endrem
\subsection{New FTC DREM}
\lab{subsec51}
%
The problem with the approach described above is that, {\em independently} from the behaviour of $\Delta(t)$, the function $w(t)$ is monotonically non-increasing and, generally, converges to zero. In this case, $\hat \theta_i^{\tt FTC}(t) \to \hat \theta_i(t)$, hence, the new estimator reduces to the standard gradient one, loosing its FTC feature. Therefore, to keep the {\em finite-time alertness} of the estimator, {\em i.e.}, to track parameter variations in finite-time upon the arrival of new excitation, it is necessary to reset the estimators \eqref{ctest} or \eqref{dtest}---a modification that is always problematic to implement. 

In the proposition below  we propose an alternative for the CT estimator of Proposition \ref{pro60} that does not suffer from this practical drawback. For the sake of brevity, we present only the derivation of a relation similar to \eqref{keyrel}, from which we can easily construct the FTC estimator.

\begin{proposition}\em
\lab{pro6}
Fix $T_{\tt D} \in \rea_{> 0}$ and define 
\begequ
\lab{dotwd}
\dot w_i^{\tt D}(t)=-\gamma_i\left[\Delta^2(t)-\Delta^2(t-T_{\tt D})\right]w_i^{\tt D}(t),\;w^{\tt D}_i(0)=1.
\endequ
Then,
$$
\left[1-w_i^{\tt D}(t)\right]\theta_i = \hatthe(t)-w_i^{\tt D}(t)\hatthe_i(t-T_{\tt D}).
$$
\end{proposition}
\begin{proof}
Without loss of generality we assume that $\Delta(t-T_{\tt D})=0$ for $t < T_{\tt D}$. Then the solution of \eqref{dotwd} is
\begequ
\lab{solwd}
w_i^{\tt D}(t)=e^{-\gamma_i\int_{t-T_{\tt D}}^t \Delta^2(s)ds}.
\endequ
Now, from the solution of the PEE \eqref{solctpee} in the interval $[t-T_{\tt D},t]$ we get
$$
\tilde \theta_i(t) = e^{-\gamma_i \int_{t - t_{\tt D}}^t \Delta^2(s)ds} \tilde \theta_i(t - T_{\tt D}).
$$
Hence, $\tilthe_i(t)=w_i^{\tt D}(t)\tilthe_i(t-T_{\tt D}).$ The proof of the claim is established rearranging the terms of the equation above.
\end{proof}

The difference between the two FTC algorithms is evident comparing $w_i(t)$ and the new signal $w_i^{\tt D}(t)$. As indicated above, the former is {\em always} non-increasing, while $w_i^{\tt D}(t)$ {\em grows} if $\Delta(t)$ increases its value in an interval of length $T_{\tt D}$, that is, if new excitation arrives to the system. In this way, the new FTC estimator preserves its FTC property if the parameters change. This fact is illustrated in the simulations of Subsection \ref{subsec62}.

\begrem
For the new FTC DREM estimator the interval excitation inequality becomes the existence of a time $t_c \geq T_{\tt D}$ such that
\begequ
\lab{inttcd}
\gamma_i \int^{t_c}_{t_c-T_{\tt D}} \Delta^2(s) ds \geq -\ln(\mu_i).
\endequ
Recalling \eqref{solwd}, it has the same interpretation as \eqref{inttc}.
\endrem  
\begrem
The choice of the coefficients $\mu_i$ is, clearly, a compromise between high-gain injection---if it is close to 1---and the time where FTC is achieved. See \cite{GERetalsysid} for additional details on this aspect.
\endrem
%
\section{Simulations}
\lab{sec6}
%
In this section we present simulations illustrating the results of Propositions \ref{pro5}, \ref{pro60} and \ref{pro6}.
\subsection{Transient performance improvement of Proposition \ref{pro5}}
\lab{subsec61}

To illustrate the performance improvement using the time-varying  term $d(t)$ introduced in Proposition \ref{pro5}, we consider the problem of parameter estimation of the CT, first-order, LTI plant described by 
\begali{
\lab{lti1}
\dot y(t)=a y(t)+bu(t),
}
where $u(t),y(t)\in\mathbb{R}$ are measurable signals and $a,b \in \rea$ are uncertain parameters that should be estimated.

Following the standard LTI systems identification procedure \cite{SASBOD}, we first re-parametrize the model \eqref{lti1} to obtain the LRE \eqref{lre}. For, we apply the filters $1\over p+\lambda$, with some $\lambda>0$, to \eqref{lti1} to get the LRE \eqref{lre} with
\begalis{
\phi(t):=\begmat{{1\over p+\lambda}[y](t) \\ {1\over p+\lambda}[u](t)},\;\theta:=\begmat{a+\lambda \\ b}.
}
Two simulation scenarios have been considered: with a plant input that is sufficiently rich or not---that is, when the regressor $\phi(t)$ is PE or not. More precisely, we considered $u(t)=15\sin(2.5t+1)$ and $u(t)=15$, respectively. For these two scenarios, we compare three different estimation schemes, namely, the standard gradient-descent \eqref{graest}, and the DREM scheme \eqref{ctest} with the operators $\calh_i$ defined in \eqref{genhi} for $d=0$ and $d(t)$ given by \eqref{d}. 

The following simulation parameters are used: $a=-0.4, b=0.4, \lambda=5, \gamma=1$, with the coefficients
\begalis{
\calh_1 &:\; n_1=1,\; A_1=-1, \ b_1=1, \ c_1 =1, \  \mu_1=0\\
\calh_2 & :\; n_2=1,\; A_2=-2, \ b_2=2, \ c_2 =1, \  \mu_2=0,
}
for the LTI part of the operators $\calh_i$.

The transient behavior of the parameter estimation errors $\tilde \theta_1(t)$ and  $\tilde \theta_2(t)$, for the three aforementioned estimators,  is shown in Figs. \ref{fig_1} and \ref{fig_2}. As predicted by the theory the gradient scheme yields a consistent estimate only for the case of sufficiently rich input, showing a significant steady state error for the constant plant input. On the other hand, both DREM schemes yield consistent estimates in both scenarios. Moreover, as expected from the analysis of Proposition \ref{pro5}, the addition of the feedforward term  $d(t)$ given in \eqref{d}, significantly improves the transient performance---achieving parameter convergence in less than a second, while the DREM scheme with $d=0$ takes almost two seconds to converge. It should also be mentioned that both DREM schemes significantly outperform the standard gradient, even in the presence of a sufficiently rich input. This property stems from the fact that, as indicated in Proposition \ref{pro2}, DREM ensures monotonicity of each element of the parameter error vector, a fact that is clearly illustrated in the simulations. 

\begin{figure}[tb]
\centering
\subfloat[][Parameter estimation error $\tilde\theta_1(t)$]{{\label{fig_2a}}\includegraphics[height=4cm]{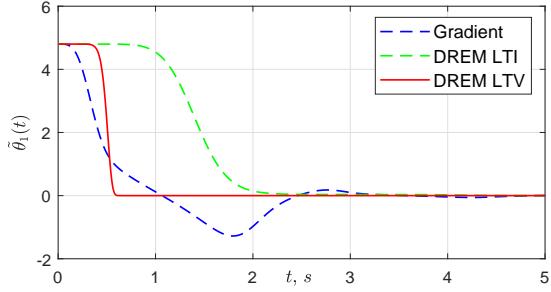}}
\\
\subfloat[][Parameter estimation error $\tilde\theta_2(t)$]{{\label{fig_2b}}\includegraphics[height=4cm]{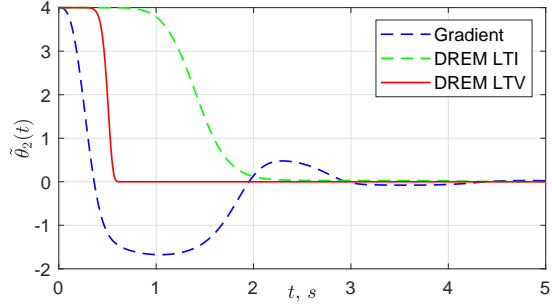}}
\vspace{-1mm}
\caption{Transients of the parameter estimation errors for different estimators and the control input $u(t)=15\sin(2.5t+1)$.}
\vspace{-3mm}
\label{fig_1}
\end{figure}
\begin{figure}[tb]
\centering
\vspace{-3mm}
\subfloat[][Parameter estimation error $\tilde\theta_1(t)$]{{\label{fig_1a}}\includegraphics[height=4cm]{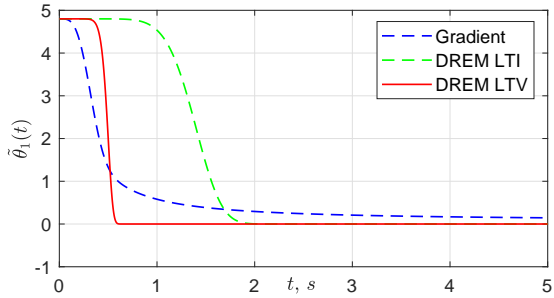}}
\\
\subfloat[][Parameter estimation error $\tilde\theta_2(t)$]{{\label{fig_1b}}\includegraphics[height=4cm]{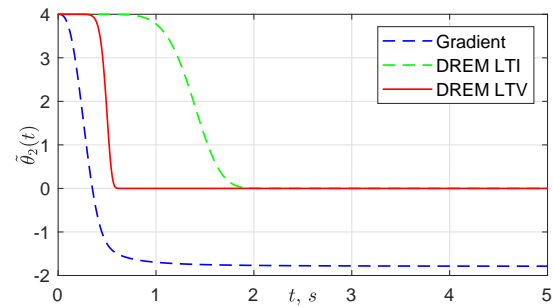}}
\vspace{-1mm}
\caption{Transients of the parameter estimation errors for different estimators and the control input $u(t)=15$.}
\label{fig_2}
\end{figure}
\subsection{Alertness preserving DREM with FTC of Proposition \ref{pro6}}
\lab{subsec62}
%
In this subsection we compare the two FTC DREMs  presented in Section \ref{sec5}. Namely, the FTC DREM of  Proposition \ref{pro60}, defined by \eqref{hatthew}, \eqref{wi}, and the new FTC DREM of Proposition \ref{pro6} given by \eqref{dotwd} and
$$
\hat \theta_i^{\tt FTC-D}(t) :={1  \over 1 -w^{\tt D}_i(t)}[\hat \theta_i(t) - w^{\tt D}_i(t)\hat \theta_i(0)],
$$
which is computed as soon as $w^{\tt D}_i(t) < \mu_i$. The objective of the simulation is to prove that the new FTC DREM is able to track time-varying parameters when new excitation arrives. This is in contrast with the old FTC DREM estimator that, since $w(t) \to 0$, converges to the gradient estimator and loses its FTC alertness property. 

We consider the simplest case of a scalar system $y(t)=\Delta(t)\theta$ and simulate the gradient estimator \eqref{ctest}, that is,
$$
\dot {\hat \theta}(t)=\gamma \Delta(t)[y(t)- \Delta(t)\hat\theta(t)],
$$
together with \eqref{wi} and \eqref{dotwd}, which are computed for $t \geq t_c$, with $t_c$ defined via the interval excitation criteria \eqref{inttc} and \eqref{inttcd}, respectively.  

We consider two scenarios: with and without excitation in $\Delta(t)$. For the first case we consider the PE signal $\Delta(t) = \sin(2\pi t)$, and for the second one $\Delta(t) = \frac{1}{t+1}$. Note that in the second case $\Delta(t)\to 0$, hence it is not PE. However, $\Delta(t) \not \in \mathcal{L}_2$, hence it satisfies the conditions for convergence of the DREM estimator. 

For simulations we set $\gamma = 2$, $\mu=0.98$, and $T_{\tt D}=0.2$. These parameters have been chosen such that the transients of both FTC estimators coincide in the ideal case when $\theta$ is constant and the system is excited. To illustrate the FTC tracking capabilities of the estimators the unknown parameter $\theta$ is time-varying and given by
\[ 
	\theta(t) = \begin{cases}
		10 \text{ for } 0\le t <10, \\
		15 \text{ for } 10\le t <20, \\		
		15-0.5(t-20) \text{ for } 20\le t <30, \\		
		10 \text{ for } t>30, \\
	\end{cases}
\]
{\em i.e.}, it starts at $10$, jumps to $15$ at $t=10$, and then linearly returns to $10$.

The transient of the estimators for $t\in[0,3]$ and  $\Delta(t) = \sin(2\pi t)$ are given in Fig. \ref{fig:fig1}, where we plot the gradient estimate $\hat \theta(t)$, as well as the old and the new FTC estimates $\hat \theta^{\tt FTC}(t)$  and $\hat \theta^{\tt FTC-D}(t)$.  We observe that, as expected, both FTC estimators are overlapped and converge in finite time, while the gradient converges only asymptotically.
  
The behavior of the estimators for $t \in [9,40]$ is shown in Figure \ref{fig:fig2}, where we also plot the time-varying parameter $\theta(t)$. As predicted by the theory, the old FTC behaves as the gradient estimator and their trajectories coincide. On the other hand, the new estimator preserves FTC alertness after the first parameter jump and achieves fast tracking of the linearly time-varying $\theta(t)$.

For the non-PE case of $\Delta(t) = \frac{1}{t+1}$, the transients of the estimators are given in Fig. \ref{fig:fig3}. We observe that both FTC estimators, again, essentially coincide in the first few seconds and converge in finite time, while the gradient does it only asymptotically. After the first parameter change at $t=10$ the old FTC and the gradient coincide, while the new FTC manages to track in finite time the parameter jump. However, during the ramp parameter change---because of the lack of excitation---neither one of the estimators can track the parameter variation but the new FTC estimator performs much better.

\begin{figure}[tb]
	\centering
	\includegraphics[height=4cm]{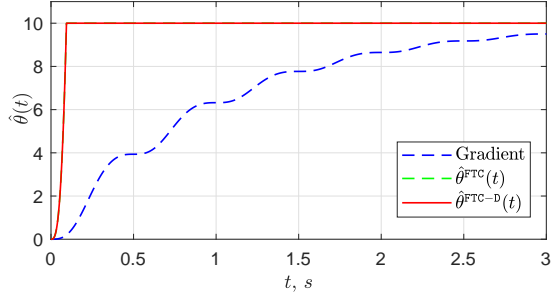}
	\caption{Transients of the parameter estimates for $t \in [0,3]$ with $\Delta(t) \in PE$.}
	\label{fig:fig1}
\end{figure}
\begin{figure}[tb]
	\centering
	\includegraphics[height=4cm]{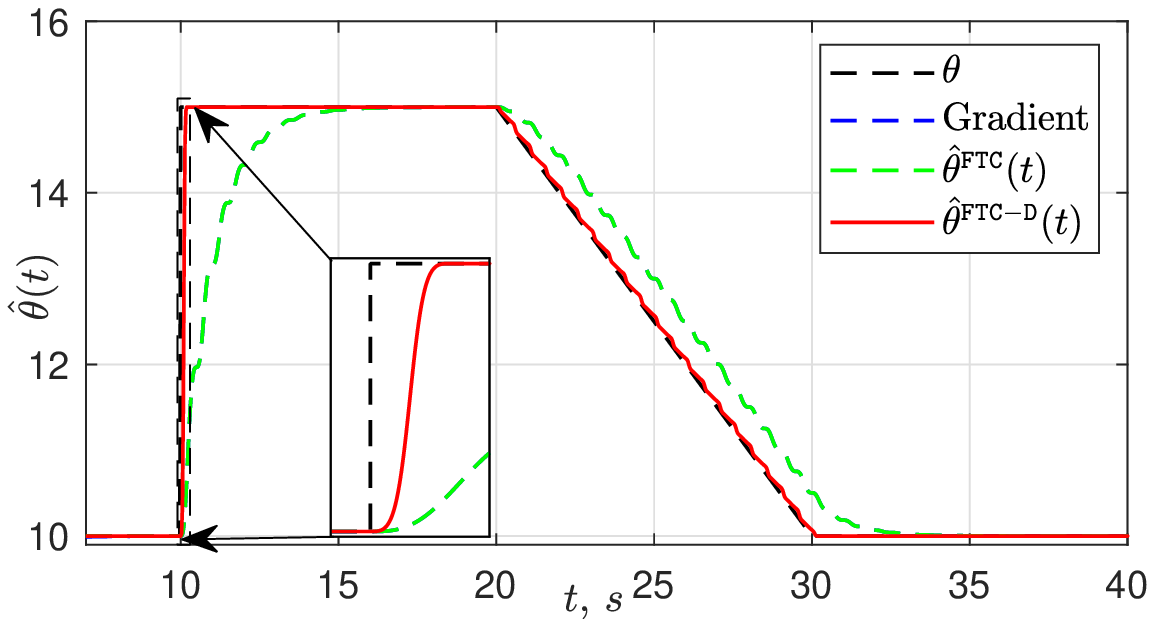}
	\caption{Transients of the parameter estimates for $t \in [9,40]$ with $\Delta(t) \in PE$.}
	\label{fig:fig2}
\end{figure}
\begin{figure}[tb]
	\centering 
	\includegraphics[height=4cm]{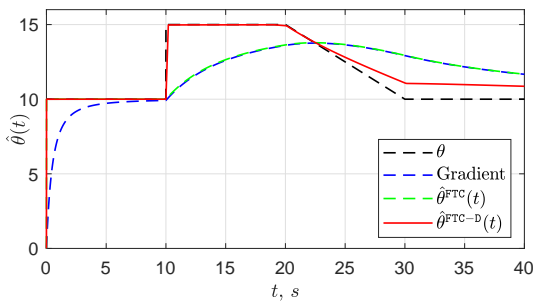}
	\caption{Transients of the parameter estimates for $\Delta(t)=\frac{1}{t+1}$.}
	\label{fig:fig3}
\end{figure}
%
\section{Future Work}
\lab{sec7}
%
Current research is under way to derive some of the new results presented only for the CT time case, to the practically important, DT case. Moreover, in the spirit of \cite{BELetal}, we are further exploring the role of the operator $\calh$ on the determinant of the extended regressor matrix $\Phi$ and we plan to study the effect of an additive signal in the LRE \eqref{lre}, to study its input-to-state stability properties. 

A  widely open, long-term research topic is how to deal with {\em nonlinear parameterizations}, that is, the case in which \eqref{lre} is replaced by $y=F(\phi,\theta)$, where $F(\cdot,\cdot)$ is a nonlinear function. Some preliminary results exploiting convexity, concavity or monotonicity may be found in \cite{ANNetal,LIUetal,LIUetalscl}. As pointed out in \cite{ARAetaltac}, DREM is directly applicable---without overparameterization---in the simplest case of separable nonlinearities, that is, when the regression is of the form $y=F_\phi(\phi)F_\theta(\theta)$. The more general case is a challenging open problem.     
%
\section*{Acknowledgment}
The authors would like to thank Vladimir Nikiforov and Dmitry Gerasimov for many useful discussions that helped us to improve the quality of our contribution. 

This paper is partly supported by by Government of Russian Federation (GOSZADANIE 2.8878.2017/8.9, grant 08-08), the European Union's Horizon 2020 Research and Innovation Programme under Grant 739551 (KIOS CoE).
%
%

\end{document}